\documentclass[reqno,11pt]{amsart}

\usepackage{amssymb,amsthm}
\usepackage{amsmath}
\usepackage{amsfonts}
\usepackage{dsfont}
\usepackage[a4paper,  margin=3.2cm]{geometry}
\makeatletter\@addtoreset{equation}{section}\makeatother

\newtheorem{theorem}{Theorem}[section]

\newtheorem{lemma}[theorem]{Lemma}

\newtheorem{proposition}[theorem]{Proposition}

\newtheorem{definition}[theorem]{Definition}
\theoremstyle{remark}

\numberwithin{equation}{section}

\title[The Lee-Yang property of isotropic vector ferromagnets]{The Lee-Yang property of  isotropic vector ferromagnets and lattice fields}
\author{Yuri Kozitsky}

\address{Instytut Informatyki i Matematyki, Uniwersytet Marii Curie-Sk{\l}odowskiej, plac Marii Curie-sk{\l}odowskiej 1, 20-031 Lublin, Poland; fax +48 81 537 5471}
\email{jurij.kozicki@mail.umcs.pl}

\begin{document}


\begin{abstract}
The Lee-Yang property of a given spin model means that its partition function has purely imaginary zeros as a function of an external magnetic field. A similar property is also used in the theory of quantum anharmonic crystals and quantum lattice fields. A number of powerful analytic methods of the mathematical theory of such models employ this property. Its suitable generalization is used in the theory of models of isotropic $D$-dimensional spins (rotors) or $D$-component quantum lattice fields. So far, the (generalized) Lee-Yang property has been established only for two-dimensional isotropic models. In this work, we prove that isotropic spin and field models living on $\mathds{Z}$ have this property for all even $D$.   

\end{abstract}

\maketitle

\section{Introduction}

The present research was inspired by an item on the list of open problems in mathematical physics recently published by Barry Simon in \cite{Simon}. In Conjecture 9.2.27 on page 570, it is proposed to ``Prove a Lee–Yang theorem for isotropic classical $D$-rotors for $D \geq 4$". Here we prove this statement for all even $D$ and a particular models of such `rotors' and  Euclidean quantum lattice fields.    

In 1952, two physicists, T. D. Lee and C. N. Yang, proved a statement \cite{LY}, which turned out to be seminal for both mathematical physics \cite{Albev,DN,Frol,Glimm,LS,Newman,Simon} and pure mathematics \cite{Borcea,Borcea1,Br,Newman}. Their original result describing the Ising spin model can be sketched as follows. Consider the variables $\sigma_l$, $l=1, \dots , N$ -- spins -- taking values $\pm 1$, and the function
\begin{equation}
\label{0}
Z^{\rm Ising}_N(h) = \sum_{\dots \sigma_l=\pm 1\dots  } \exp\left(h \sum_{1\leq l\leq N}\sigma_l + \sum_{1\leq l < l' \leq N} J_{ll'} \sigma_l \sigma_{l'} \right).
\end{equation}
In statistical physics, it is the partition function of a magnet, where $J_{ll'}$ are interaction intensities and $h$ stands for an external magnetic field. 
As a finite sum of exponentials, it can be continued to an exponential-type entire function of $h\in \mathds{C}$. According to the Lee-Yang theorem, all zeros of $Z^{\rm Ising}_N$ lie on the imaginary axis, provided that the interaction is of ferromagnetic type, i.e., $J_{ll'}\geq 0$ for all $l,l'$. Later, this result was extended in different directions; see \cite{DN,Frol,Glimm,LS,Newman}, \cite[page 171]{Albev}, and \cite[Chapter 3]{Simon}, with applications in statistical physics and lattice Euclidean quantum field theory. In particular, the summation in \eqref{0} was replaced by integration, and the variables $\sigma_l$ were turned into $D$-vectors. For $D=2$, an analog of the Lee-Yang theorem was proved in \cite{LS}, see also \cite[Chapter 3]{Simon}.  

In this article, we study the function
\begin{equation}
\label{1}
Z_N (z) = \int \exp\left( \sum_{l=1}^N z \cdot \sigma_l + J\sum_{l=1}^{N-1} \sigma_l \cdot \sigma_{l+1}\right) \prod_{l=1}^N \chi(d\sigma_l). 
\end{equation}
In the language of statistical physics, the corresponding model lives on the graph $\mathds{Z}$, which means that $J_{ll'} = J>0$ if $|l-l|=1$, and $J_{ll'} =0$ otherwise. In \eqref{1}, the integral is taken over $(\mathds{R}^{D})^N$, $D,N\in \mathds{N}$, 
$z=x+iy\in \mathds{C}^D$,  $z\cdot \sigma = x\cdot \sigma + i y\cdot \sigma$ is the corresponding scalar product  
and $\chi$ is a finite positive measure on $\mathds{R}^D$ -- single-spin measure -- satisfying certain conditions. In particular, we assume that 
\begin{equation}
\label{3}
\forall a>0 \qquad \int_{\mathds{R}^D} e^{a\sigma^2} \chi(d\sigma) < \infty, \qquad \sigma^2:= \sigma\cdot \sigma,
\end{equation}
and that $\chi$ is \emph{isotropic}, i.e., it satisfies $\chi (UA) = \chi(A)$, holding for all orthogonal transforms $U\in O(D)$ and each Borel  $A\subset \mathds{R}^D$. An example is the measure supported on the sphere $\mathds{S}_r \subset \mathds{R}^D$ of radius $r>0$, the restriction of which to this sphere is a uniform measure.  It corresponds to the aforementioned model of isotropic rotors. Our result is the statement that, for all even $D\in \mathds{N}$, the function in \eqref{1} can be written in the form
\begin{equation}
\label{1a}
Z_N (z) = Z_N (0) \prod_{j=1}^\infty (1 + \gamma_{j,N} z^2),\quad z^2 :=z\cdot z, \quad \gamma_{j,N} >0, \quad \sum_{j} \gamma_{j,N} < \infty,
\end{equation}
whenever $\chi$ belongs to a subset of the set of isotropic measures on $\mathds{R}^D$, for which \eqref{3} and the following condition are verified
\begin{equation}
\label{4}
\widehat{\chi} (z):=
\int_{\mathds{R}^D} e^{z\cdot \sigma} \chi(d\sigma) = C_\chi \prod_{j=1}^\infty (1 + \gamma_j z^2), \quad \gamma_{j} >0, \quad \sum_{j} \gamma_{j} < \infty. 
\end{equation}
The function $\widehat{\chi}$ will be referred to as the Laplace transform of $\chi$.
The function of $z^2$ that appears on the right-hand side of \eqref{4} is a Laguerre entire function of the first kind. These are polynomials possessing real non-positive zeros only, or limits of such polynomials, taken in the topology of uniform convergence on compact subsets of $\mathds{C}$; see \cite[pages 8--23]{Iliev}. 
In the following, by $\mathcal{L}$ we denote the set of all such functions. Each $f\in \mathcal{L}$ can be written in the form
\begin{equation}
\label{4a}
f(\zeta) = C \zeta^m e^{\alpha \zeta} \prod_{j=1}^\infty (1 + \gamma_j \zeta),
\end{equation}
with $m\in \mathds{N}_0$, $\alpha \geq 0$, and the set $\{\gamma_j\}$ (infinite, finite, or empty) that satisfies the condition in \eqref{4}. In particular,  $e^{\alpha \zeta} \in \mathcal{L}$, for all $\alpha\geq 0$. Noteworthy, for $D=1$, $\varphi_\chi(t):=\widehat{\chi}(it)$ is the Fourier transform of $\chi$, which by \eqref{4} has only real zeros. The study of such functions -- Riemannschen ganzen Funktionen -- goes back to Riemann; see \cite[Chapter 3]{Iliev}. In particular, verifying the famous Riemann hypothesis concerning the zeros of the zeta function amounts to verifying this property for a certain measure; see \cite{Newm}.

\section{The Result}
In this section, we first introduce the set of single-spin measures $\chi$, which we are going to deal with. 
Then we formulate and discuss our result: the validity of the representation \eqref{1a} under the condition that the single-spin measure $\chi$ in \eqref{1} belongs to this set.  

\subsection{Strongly isotropic measures}

Let $\mathcal{S}(\mathds{R}^D)$ and $\mathcal{S}(\mathds{R}_{+})$ be the Schwartz spaces of real-valued test functions defined on $\mathds{R}^D$ and $\mathds{R}_{+}:= [0,+\infty)$, respectively. Let also $O(D)$ denote the group of orthogonal transforms of $\mathds{R}^D$. For $U\in O(D)$ and $f\in \mathcal{S}(\mathds{R}^D)$, we define $f_U$ by the equality $f_U(x) = f(Ux)$.  Set $\mathcal{S}_D=\{f\in \mathcal{S}(\mathds{R}^D): \forall U \in O(D) \  f_U = f \}$. It is known, see \cite[Theorem 2.1 and Note on page 245]{Bleher}, that for each 
$f\in \mathcal{S}_D$, there exists a unique $\phi \in \mathcal{S}(\mathds{R}_{+})$ such that $f(x)=\phi (x^2)$. Moreover, the map $f \mapsto \phi$ is an isomorphism of the corresponding Schwartz spaces. A distribution, $T\in \mathcal{S}'(\mathds{R}^D)$, is said to be isotropic if 
\[
\forall U\in O(D) \quad T(f_U) = T(f):= \int_{\mathds{R}^D} T(x) f(x) dx.
\]
By $\mathcal{S}'_D$ we denote the set of all isotropic $T\in \mathcal{S}'(\mathds{R}^D)$. 
It is known that $\mathcal{S}(\mathds{R}_{+})$ is weakly dense in $\mathcal{S}'(\mathds{R}_{+})$. 
Let $\mathcal{T}$ be the strongest topology on $\mathcal{S}'(\mathds{R}^D)$, for which $\phi \mapsto f$ can be extended to a continuous map from $\mathcal{S}'(\mathds{R}_{+})$ to $\mathcal{S}'(\mathds{R}^D)$. By $\overline{\mathcal{S}_D}$ we denote the closure of $\mathcal{S}_D$ in $\mathcal{T}$. Then, for each $T\in \overline{\mathcal{S}_D}$, one can find $\tau \in \mathcal{S}'(\mathds{R}_{+})$ such that $T(x) = \tau (x^2)$. Obviously, $\overline{\mathcal{S}_D} \subset \mathcal{S}'_D$.  
\begin{definition}
\label{1df}
A positive measure, $\chi$, on $\mathds{R}^D$ is said to be strongly isotropic if: (a) it satisfies \eqref{3}; (b) there exists $T\in \overline{\mathcal{S}_D}$ such that
\begin{equation}
\label{2}
\chi(d\sigma) = T(\sigma) d \sigma = \tau(\sigma^2) d \sigma.
\end{equation}
Furthermore, we say that an isotropic measure has the Lee-Yang property if it satisfies \eqref{4}. The latter is related to all 
isotropic measures.
\end{definition}
The first relevant example is the aforementioned uniform measure on the sphere $\mathds{S}_r \subset \mathds{R}^D$, for which $T(\sigma)= \delta_r(\sigma)= \delta (\sigma^2 - r)$, $r>0$, $\delta$ being Dirac's function. Its Laplace transform is
\begin{gather}
 \label{5}
 \widehat{\chi} (z) = \int e^{z\cdot \sigma} \delta (\sigma^2 - r) d\sigma = \pi^{D/2} w_D (z^2, r), \\[.2cm] \nonumber w_D (\zeta, r) = \sum_{n=0}^\infty \frac{\zeta^n r^{D/2 + n -1}}{2^{2n}n! \Gamma(D/2 + n) }, 
 \end{gather}
where $\Gamma$ is Euler's $\Gamma$-function. By means of the P\'olya-Schur theorem, see \cite[pages 16--23]{Iliev}, it is possible to show that $w_D\in \mathcal{L}$ for all $D$. Moreover, the Laplace transform of each strongly isotropic measure, see \eqref{2}, can be presented in the form
\begin{equation}
\label{5a}
\widehat{\chi} (z) = \pi^{D/2}\int_0^{+\infty} w_D (z^2, r) \tau (r) dr.
\end{equation}
By \eqref{5}, one readily gets
\begin{equation}
\label{6}
\frac{\partial}{\partial_\zeta } w_{D} (\zeta, r) = 2^{-2} w_{D+2} (\zeta, r),
\end{equation}
which will be used throughout. It is a particular case of the following property
\begin{equation}
 \label{L}   \forall f\in \mathcal{L} \ \ f' \in \mathcal{L};
\end{equation}
i.e., differentiation maps $\mathcal{L}$ into itself, see \cite[Proposition 2.6]{Koz1}. In view of this, we define
\begin{equation}
 \label{6b}
 \mathcal{L}^{(s)} = \{ f\in \mathcal{L}: f^{(s)}\in \mathcal{L}\}, \qquad s\in \mathds{N}.
\end{equation}
Clearly, $\mathcal{L}^{(s+1)}\subset \mathcal{L}^{(s)} \subset \mathcal{L}$. 

Let $\tau \in \mathcal{S}'(\mathds{R}_{+})$ be such that the following holds, cf. \eqref{3},
\begin{equation}
 \label{3a}
 \forall a >0 \qquad \int_0^{+\infty} e^{ar} \tau (r) dr < \infty. 
\end{equation}
For this $\tau$, we set, cf. \eqref{5a}, 
\begin{eqnarray}
\label{3b}
v_\tau (\zeta, D) = \pi^{D/2}\int_0^{+\infty}w_D (\zeta, r) \tau(r) d r, \qquad \zeta \in \mathds{C},
\end{eqnarray}
which is an entire function of order less than one, or of order one and of minimal type. 
Let us notice that each $\tau$ satisfying \eqref{3a} by \eqref{2} and
\eqref{3b} defines a family of measures and their Laplace transforms.  
For obvious reason, by \eqref{6} we have that 
\begin{equation}
\label{3c}
v_\tau(\zeta , D+2m) = (4 \pi)^m v_\tau^{(m)} (\zeta, D), \qquad m \in \mathds{N}.  
\end{equation}
By this formula and \eqref{6b}, one immediately gets the proof of the following statement.
\begin{lemma}
\label{1lm}
For some $D\in \mathds{N}$, let $v_\tau(\cdot, D)\in \mathcal{L}$. Then 
$v_\tau(\cdot, D+2m)\in \mathcal{L}^{(m)}$ for all $m\in \mathds{N}_0$.
\end{lemma}
As mentioned above, the set $\overline{\mathcal{S}_D}$ contains test functions. In the next important example, $T$ is such a function. 
\begin{proposition} \cite[Theorem 3.1]{Koz}
\label{1pn}
Let $T$ be of the following form
\begin{equation}
\label{7}
 T(\sigma) =: \tau(\sigma^2)= f(\sigma^2) \exp\left( - g (\sigma^2\right)),   
\end{equation}
where both $f$ and $g$ are positive. Moreover, assume that
$f$ and the derivative $g'$ are in $\mathcal{L}$. If $g$ is a polynomial, then its degree should be at least two. Then the measure as in \eqref{2} with this $T$ has the Lee-Yang property for all $D\in \mathds{N}$.  In particular, $v_\tau (\cdot , D)\in \mathcal{L}$ with $\tau$ as 
in \eqref{7}.
\end{proposition}
Taking in \eqref{7} $f(\zeta)=e^{a''\zeta}$, $a''\geq 0$, see \eqref{4a}, and $g (\zeta) = a' \zeta + b\zeta^2$, $a''\geq 0$, 
$b>0$, one concludes that the measure $\chi(d\sigma)= \exp(- a\sigma^2 - b(\sigma^2)^2) d\sigma$, $a= a' - a''$, has   
the Lee-Yang property for all $D$ and $a\in \mathds{R}$. This is a generalization to all $D$ of the well-known fact concerning the $\phi^4$ quantum lattice field, see \cite[Theorem 3.4.8, page 233]{Simon} and \cite{DN,Glimm}. Another example of this kind is the measure
\[
\chi(d\sigma) = \exp\left(- a \sigma^2 - b(\sigma^2)^2 - c(\sigma^2)^3\right)d\sigma, \quad b, c>0 , \ \  a\leq b^2/3c.
\]
According to Proposition \ref{1pn}, it has the Lee-Yang property for all $D$. At the same time, it is known \cite[page 71]{Glimm}, that, for $D=1$, the measure 
\[
\chi(d\sigma) = \exp(-f(\sigma^2)) d \sigma = \exp\left(- a \sigma^2 - \sigma^2(\sigma^2-1)^2\right)d\sigma, 
\]
does not have the Lee-Yang property for certain values of $a$. At the same time, for this $f$, we have $f' (\zeta)= 3\zeta^2 - 4 \zeta +a +1$, both zeros of which do not lie on the imaginary axis for all real $a$.

\subsection{The statement}
We are now ready to formulate and discuss our result.
\begin{theorem}
\label{1tm}
Let a strongly isotropic measure, $\chi$, possess the Lee-Yang property, see Definition \ref{1df}. Then, for every $J>0$ and even integer $D$, the function \eqref{1} can be written as in \eqref{1a}. 
\end{theorem}
The result just stated can be interpreted as follows. Let random $D$-dimensional vectors $S_1, \dots, S_N$ have the joint probability distribution given by the probability measure
\[
\nu_N (d\sigma_1, \dots , d\sigma_N)  = \frac{1}{Z_N(0)} \exp\left( J\sum_{l=1}^{N-1} \sigma_l \cdot \sigma_{l+1}\right) \prod_{l=1}^N \chi(d\sigma_l),
\]
where $Z_N$, $J$, and $\chi$ are as in \eqref{1}. Let also $\nu$ be the probability measure on $\mathds{R}^D$ which is the law of $S= S_1 + \cdots +S_N$. By Theorem \eqref{1tm}, it follows that its Laplace transform is
\begin{equation*}
\widehat{\nu} (z) = \prod_{j=1}^\infty(1+\gamma_{j,N}z^2),
\end{equation*}
where $\gamma_{j,N}$ are as in \eqref{1a}. That is, we state that $\nu$ has the Lee-Yang property.

The proof of Theorem \ref{1tm} is essentially based on the statement proved by Lieb and Sokal in \cite{LS}. 
We present it here in the form adapted to the context. Set
\begin{equation}
 \label{8}
 L= \{z=x+i y \in \mathds{C}^2: \exists u\in \mathds{R}^2 \  (x\cdot u)^2 + (y\cdot u)^2 > y^2 u^2\},
\end{equation}
and 
\begin{equation}
\label{9}
Q_N (z, {\sf J}) = \int \exp\left( \sum_{1\leq l\leq N}z_l\cdot \sigma_l + \sum_{1\leq l < l' \leq N} J_{ll'} \sigma_l \cdot \sigma_{l'} \right) \mu_N(d\sigma_1, \dots , d \sigma_N),
\end{equation}
where ${\sf J}=(J_{ll'})$ and the integral is taken over $(\mathds{R}^2)^N$. Note that the variables $\sigma_l$ in \eqref{9} are two-dimensional. In the statement below, by writing ${\sf J}=0$ we mean $J_{ll'}=0$ for all $l,l'$. 
\begin{proposition} \cite[Theorem 4.3 and Remark 1]{LS}
\label{2pn}
Assume that the positive measure $\mu_N$ in \eqref{9} is such that: (a) the integral exists for all $J_{ll'}\geq 0$; (b)  $Q_N (z, 0) \neq 0$ whenever $z\in L^N$, i.e., whenever all $z_l$ are in $L$. Then $Q_N (z, {\sf J}) \neq 0$  for each $z\in L^N$ provided $J_{ll'} \geq 0$ for all $l,l'$.  
\end{proposition}
First, we mention that this statement is about a function of $N$ complex variables. As we show below, $Z_N (h)$, obtained from $Q_N$ by setting all $z_l$ equal $h$, can be presented in the form \eqref{1a}.  
Then our Theorem \ref{1tm} generalizes the corresponding result of \cite{LS} to all even $D$, but for a particular choice of the matrix $(J_{ll'})$. Here, we remark that the numerical results obtained in \cite[Section 2]{Kurtze} point to the validity of our statement for $\chi$ as in \eqref{5} and all $D$. Another particular choice of the aforementioned matrix $\sf J$ is $J_{ll'}= (1+d_{ll'})^{-(1+a)}$, where $d_{ll'}$ is a \emph{hierarchical} distance on $\{1, \dots , N\}$; see, e.g., \cite{BM}. In this case, the validity of the representation \eqref{1a} was proved for all $D\in \mathds{N}$ and all isotropic measures possessing the Lee-Yang property, see \cite{Koz1,Koz2}.

\section{The Proof}

In the sequel, we use the following notations
\begin{equation}
\label{C}
\mathds{C}^{-} = \mathds{C}\setminus (-\infty, 0], \qquad \mathds{C}^{-}_2 =   
\mathds{C}^{-}\times \mathds{C}^{-} \times \mathds{C}.   
\end{equation}
For $z=x+i y\in \mathds{C}^2$, we let $\ell (z) = z^2$. Then $\ell: \mathds{C}^2 \to \mathds{C}$ is obviously surjective. Moreover, 
\begin{equation}
 \label{10}
 \ell^{-1} (\mathds{C}^{-}) = L,
\end{equation}
see \eqref{8}. Indeed, take any $\zeta= \xi + i \eta \in \mathds{C}^{-}$ and find $z\in \mathds{C}^2$ such that $\zeta = \ell(z)$, i.e., such that $x^2 - y^2 = \xi$ and $2x\cdot y = \eta$. If $\eta\neq 0$, 
we take any $x, y\neq 0$ such that $x\cdot y = \eta/2$, and then set $u=y$, which yields 
$(x\cdot u)^2 + (y\cdot u)^2 = \eta^2/4 + (y^2)^2 > (y^2)^2$; hence $x+iy\in L$. For $\eta=0= x\cdot y$, $\xi= x^2 -y^2$ should be strictly positive. Then
we have to take $z$ with $x^2 > y^2$. For such $z$, take $u = \alpha x + \beta y$. Then $u^2 = \alpha^2 x^2 + \beta^2 y^2$. At the same time,
\[
(x\cdot u)^2 + (y\cdot u)^2 = \alpha^2 (x^2)^2 +  \beta^2 (y^2)^2 > \alpha^2 x^2 y^2 +  \beta^2 (y^2)^2 = u^2 y^2,
\]
which yields $\ell^{-1} (\mathds{C}^{-}) \subset L$. Let us prove $\ell (L) \subset \mathds{C}^{-}$. Take any $\zeta  = \xi + i \eta\in \ell (L)$, and let $z\in L$ be its preimage, i.e., $\xi = x^2 -y^2$ and $\eta = 2 x\cdot y$. If $x\cdot y \neq 0$, then $\zeta \in \mathds{C}^{-}$. For $x\cdot y = 0$, we have the following possibilities: (a) $y=0$; (b) $y\neq 0$. For (a), $x^2>0$ since it should be $(x\cdot u)^2> 0$ for some $u$. This yields $\xi>0$, and hence $\zeta \in \mathds{C}^{-}$. For $y\neq 0$, let $u$ be as in \eqref{8}. Write $u =\alpha x+\beta y$. Then $u^2 = \alpha^2 x^2 + \beta^2 y^2$, and 
\[
\alpha^2 (x^2)^2 + \beta^2 (y^2)^2 > \alpha^2 x^2 y^2+ \beta^2 (y^2)^2, 
\]
which implies $x^2 > y^2$, and hence $\zeta \in \mathds{C}^{-}$. 

Now we define $\ell_{2,2}:\mathds{C}^2 \times \mathds{C}^2\to \mathds{C}^3$ by the formula
\begin{equation}
\label{10z} 
\ell_{2,2} (z_1, z_2) = (\zeta_1, \zeta_2, \zeta_{1,2}), \quad \zeta_{1} = z_1^2, \quad \zeta_{2}= z_2^2, \quad \zeta_{1,2}= z_1 \cdot z_2.
\end{equation}
Let us prove that
\begin{equation}
\label{10y}
\ell_{2,2} (L\times L) := \ell_{2,2} (L^2) =  \mathds{C}^{-}_2,
\end{equation}
see \eqref{C}. The inclusion 
\[
\ell_{2,2} (L^2) \subset \mathds{C}^{-}_2
\]
follows by \eqref{10} since $\zeta_k = \ell (z_k)$, $k=1,2$. To prove the opposite inclusion, take $\zeta_k = \xi_k + i \eta_k \in \mathds{C}^{-}$, $k=1,2$, and $\zeta_{1,2}= \xi_{1,2} + i \eta_{1,2}\in \mathds{C}$. We aim to find $(z_1, z_2) \in L^2$ which is a $\ell_{2,2}$-preimages of the triplet $(\zeta_1, \zeta_2, \zeta_{12})$. If both $\eta_k$ are nonzero, then $z_1, z_2\in L$, see above. Thus, it remains to consider the case of $\eta_1=0$ and $\xi_1 >0$. Then the corresponding $z_1$, $z_2$ should satisfy $x_1^2 = \xi_1 + y_1^2$, $x_1 \cdot y_1 =0$, $z_2= x_2 + i y_2 \in L$, and also 
\begin{equation}
 \label{10x}
 x_2^2 - y_2^2 = \xi_2, \quad x_2 \cdot y_2 = \eta_2/2, \quad x_1 \cdot x_2 - y_1\cdot y_2 = \xi_{1,2}, \quad x_1 \cdot y_2 + y_1\cdot x_2 = \eta_{1,2}. 
\end{equation}
Note that $z_2 = x_2 + i a x_2$ is in $L$ for each $a\in \mathds{R}$. Indeed, take $u=x_2$, and get $(u\cdot x_2)^2 + (u\cdot y_2)^2 = (1+ a^2)(x_2^2)^2 > a^2 (x_2^2)^2 = u^2 y_2^2$. For this choice of $y_2$, by the first two equations in \eqref{10x}, we have the following
\begin{equation}
\label{10u}
a = \frac{\eta_2}{\xi_2 + \sqrt{\xi_2^2 + \eta_2^2}}, \qquad x_2^2 = \frac{1}{2}(\xi_2 + \sqrt{\xi_2^2 + \eta_2^2}),
\end{equation}
and also
\begin{equation*}
x_1 \cdot x_2 = \frac{1}{1+a^2} (\xi_{1,2} + a \eta_{1,2}), \qquad y_1 \cdot x_2 = \frac{1}{1+a^2} (\eta_{1,2} - a \xi_{1,2}).
 \end{equation*}
Since $x_1\cdot y_1 = 0$, we can write $x_2 = \alpha x_1 + \beta y_1$, and then get from the latter
\[
\alpha = \frac{\xi_{1,2} + a \eta_{1,2}}{(1+a^2) (\xi_1 + y_1^2)}, \qquad \beta = \frac{\eta_{1,2} - a \xi_{1,2}}{(1+a^2) y_1^2},
\]
with $a$ given in \eqref{10u} and an arbitrary $y_1^2>0$. This includes the case of $\eta_2=0$. Now 
\[
x_2^2= \alpha^2 x_1^2 + \beta^2 y_1^2 = \alpha^2 (\xi_1 + y_1^2) + \beta^2 y_1^2 ,
\]
which should coincide with $x_2^2$ obtained in \eqref{10u}. This yields an equation for $y_1^2$. It has the following positive solution 
\begin{equation*}
 y_1^2 = \frac{1}{2x_2^2}\left[\gamma^2 + \delta^2 - \xi_1x_2^2 + \sqrt{(\gamma^2 + \delta^2 - \xi_1x_2^2)^2 + 4 \delta^2 \xi_1 x_2^2 }
\right],
\end{equation*}
with $\gamma = (\xi_{1,2} + a\eta_{1,2})/ (1+a^2)$ and $\delta = (\eta_{1,2} - a\xi_{1,2})/ (1+a^2)$. In these expressions, $a$ and $x_2^2$ are to be taken from \eqref{10u}.  This completes the proof of \eqref{10y}

Now we turn to proving Theorem \ref{1tm}, which we do by induction in $N$.  Define
\begin{eqnarray}
\label{11} F_N (z_1, z_2) = \int \exp\left(z_1 \cdot \sigma_N + \sum_{l=1}^{N-1} z_2 \cdot \sigma_l + J \sum_{l=1}^{N-1} \sigma_l \cdot \sigma_{l+1} \right) \prod_{l=1}^N \chi(d\sigma_l),    
\end{eqnarray}
where $z_1, z_2\in \mathds{C}^D$ and the integral is taken over $(\mathds{R}^D)^N$. Obviously, $F_N$ is an entire function of two complex variables. For $N\geq 3$, such functions satisfy the following recurrence
\begin{eqnarray}
\label{12}
F_N (z_1, z_2) = \bigg{[}\exp \left(J \mathcal{D}_{z_1} \cdot \mathcal{D}_{z_2} \right)\widehat{\chi} (z_1) F_{N-1} (z_2,z_3) \bigg{]}_{z_3=z_2},
\end{eqnarray}
which can be deduced from the definition in \eqref{11}. Here
\begin{equation}
\label{13}    
\mathcal{D}_{z_1} \cdot \mathcal{D}_{z_2} = \sum_{j=1}^D \frac{\partial^2}{\partial z_{1,j}\partial z_{2,j}}, \qquad z_k =(z_{k,1}, \dots , z_{k,D}) \in \mathds{C}^D, \ \ k=1,2.
\end{equation}
Differential operators, such as that in \eqref{12}, \eqref{13}, are defined using the technique developed in \cite{LS}. For $N=2$, instead of \eqref{12} we have the following formula
\begin{equation}
\label{14}
F_{2} (z_1, z_2) = \exp \left(J \mathcal{D}_{z_1} \cdot \mathcal{D}_{z_2} \right)\widehat{\chi} (z_1) \widehat{\chi} (z_2).
\end{equation}
Now we define the Gram maps $\ell_{m,D} : (\mathds{C}^D)^m \to \mathds{C}^{m(m+1)/2}$, $m=1,2,3$, by the formulas, cf. \eqref{10z},
\begin{equation*}
\ell_{m,D} (z_1, \dots z_m) = \{\zeta_{jk}: 1\leq j\leq k \leq D\}, \qquad \zeta_{jk} = z_j \cdot z_k.  
\end{equation*}
The image $\ell_{m,D}((\mathds{C}^D)^m)$ is determined by the rank of the corresponding Gram matrix $\Gamma^{m,D} = (\Gamma^{m,D}_{jk})_{m\times m}$, $\Gamma^{m,D}_{jk}= z_j \cdot z_k$. Namely, it should verify ${\rm rank}(\Gamma^{m,D}) \leq D$. Then, we have 
\begin{equation}
\label{15z}
\ell_{2,D}((\mathds{C}^D)^2) = \mathds{C}^3, \ {\rm for} \ {\rm all} \ D\geq 2, \quad {\rm and} \quad  \ell_{3,D}((\mathds{C}^D)^3)= \mathds{C}^6,  \ {\rm for} \ {\rm all} \ D\geq 4.
\end{equation}
At the same time, 
\begin{eqnarray}
\label{16}
& & \ell_{3,2}((\mathds{C}^2)^3) :=M \\ & & = \left\{ \{\zeta_{jk}\} : \zeta_1 \zeta_2 \zeta_3 + 2 \zeta_{12}\zeta_{23}\zeta_{13} - \zeta_{1}\zeta^2_{23}- \zeta_{2}\zeta^2_{13} - \zeta_{3}\zeta^2_{12}=0\right\} \subset \mathds{C}^6. \nonumber
\end{eqnarray}
By the strong isotropic property of $\chi$, see Definition \ref{1df} and \eqref{3}, the functions defined in \eqref{11}, \eqref{12}, \eqref{14}, are isotropic, i.e., $F_N(z_1,z_2)= F_N (Uz_1, Uz_2)$, 
$U\in O(D)$. According to the First Main Theorem \cite[Pages 30-33]{Weil}, and its generalizations to entire functions, it follows that there exist entire functions $\Psi_{N,D}$, defined in $\mathds{C}^3$,  such that the following holds
\begin{equation}
\label{17}
 F_N (z_1, z_2) = \Psi_{N,D} (z_1^2, z_2^2, z_1 \cdot z_2), \qquad z, z_1, z_2 \in \mathds{C}^D, \quad N\geq 3.  
\end{equation}
By \eqref{12}, \eqref{14}, and \eqref{17}, one derives the following recursion relations for these functions
\begin{gather}
 \label{18}
 \Psi_{2,D} (\zeta_1, \zeta_2, \zeta_{12}) = \exp\left( J \Delta_{2,D} \right) v (\zeta_1, D) v(\zeta_2, D), \\ \noindent \Psi_{N,D} (\zeta_1, \zeta_2, \zeta_{12}) = \left[ \exp\left( J \Delta_{3,D} \right) v(\zeta_1, D) \Psi_{N-1,D}(\zeta_2, \zeta_3, \zeta_{23}) \right]_{2=3}, \nonumber
\end{gather}
where $v=v_\tau$, see \eqref{3b}, and $2=3$ means $\zeta_{23} = \zeta_3 = \zeta_2$ and $\zeta_{13}= \zeta_{12}$. Furthermore,
\begin{eqnarray}
\label{19}
\Delta_{3,D} & = &  D \partial_{12} +  2 \zeta_1 \partial_1 \partial_{12} + 2 \zeta_2 \partial_2 \partial_{12} +  \zeta_3 \partial_{13} \partial_{23}+ \zeta_{12}(4\partial_1 \partial_2 + \partial_{12}^2 ) \\ \nonumber & + &     \zeta_{13}(2\partial_1 \partial_{23} +  \partial_{12}\partial_{13}) + \zeta_{23}(2\partial_2 \partial_{13} +  \partial_{12}\partial_{23})  , \qquad \partial_{pq} = \frac{\partial}{\partial \zeta_{pq}},
\end{eqnarray}
with the convention $\zeta_{pp}=\zeta_p$ and $\zeta_{pq}= \zeta_{qp}$ for $p>q$. And also
\begin{equation}
\label{20}
\Delta_{2,D} = D \partial_{12} + 2 \zeta_1 \partial_1 \partial_{12} + 2 \zeta_2 \partial_2 \partial_{12} + \zeta_{12}(4 \partial_1 \partial_{2} + \partial_{12}^2).
\end{equation}
Both operators $\Delta_{k,D}$, $k=2,3$,  were calculated to satisfy the conditions
\begin{eqnarray}
\label{20a}
\Delta_{2,D} F(\zeta_1, \zeta_2, \zeta_{12})& =  & \mathcal{D}_1 \cdot \mathcal{D}_2 F(z_1^2, z_2^2, z_1 \cdot z_2), \\ \nonumber
\Delta_{3,D} G(\zeta_1, \zeta_2, \zeta_3, \zeta_{12}, \zeta_{13}, \zeta_{23})& =  & \mathcal{D}_1 \cdot \mathcal{D}_2 G(z_1^2, z_2^2, z_3^2, z_1 \cdot z_2, z_1\cdot z_3 , z_2 \cdot z_3),
\end{eqnarray}
where $F$ and $G$ are appropriate functions, $\mathcal{D}_1 \cdot \mathcal{D}_2 = \mathcal{D}_{z_1} \cdot \mathcal{D}_{z_2}$ is as in \eqref{13}, and $\zeta_{jk} = z_j \cdot z_k$. 

In both cases $k=2,3$ in \eqref{19}, \eqref{20}, the only term in $\Delta_{k,D}$ that contains $\zeta_1$ is $2\zeta_1 \partial_1 \partial_{12}$. Keeping in mind that 
\[
\zeta_1 \partial_1 \partial_{12} \partial_1 = \partial_1 \zeta_1 \partial_1 \partial_{12}  - \partial_1\partial_{12},
\]
we obtain
\begin{equation}
 \label{21}
 \partial_1 \Delta_{k,D-2} = \Delta_{k,D} \partial_1  ,
\end{equation}
which we use to reduce the case of a given even $D$ to $D=2$.

Now we prove the lemma that is the main ingredient of the proof of Theorem \ref{1tm}.
\begin{lemma}
\label{5lm}
For all integer $N\geq 2$ and even $D$, it follows that $\Psi_{N,D}(\zeta_1, \zeta_2, \zeta_{12})\neq 0$ 
whenever $\zeta_1, \zeta_2 \in \mathds{C}^{-}$. 
\end{lemma}
\begin{proof}
The proof will be carried out by induction in $N$. Thus, we start by considering $N=2$, see \eqref{18} and \eqref{20}. If $D=2$, for $\zeta_1, \zeta_2 \in \mathds{C}^{-}$, and any $\zeta_{12}$, by \eqref{14}, \eqref{10y}, and Proposition \ref{2pn} we have 
\begin{equation*}
\Psi_{2,2} (\zeta_1, \zeta_2 , \zeta_{12})  = F_2 (z_1, z_2) \neq 0,
\end{equation*}
where $z_1,z_2 \in L$ are such that $\ell_{2,2} (z_1, z_2 ) = (\zeta_1, \zeta_2 , \zeta_{12})$. Indeed, by \eqref{9}, \eqref{14}, and \eqref{20a}, 
it follows that
\begin{gather}
\label{22a}
    F_2 (z_1, z_2) = Q_2 (z_1, z_2 ,{\sf J}) =\exp(J \mathcal{D}_1 \cdot \mathcal{D}_2) Q_2 (z_1, z_2, 0), \\ Q_2 (z_1, z_2, 0) = 
\widehat{\chi}(z_1) \widehat{\chi}(z_2). \nonumber
\end{gather}
Let now $D= 2+ 2m$. 
By \eqref{3c} and Lemma \ref{1lm}, and then by \eqref{21}, we get
\begin{gather}
\label{23}
    \Psi_{2,D} (\zeta_1, \zeta_2, \zeta_{12}) = (4\pi \partial_1)^m \Phi_2  (\zeta_1, \zeta_2, \zeta_{12}), \\ \Phi_2  (\zeta_1, \zeta_2, \zeta_{12}) := \exp(J\Delta_{2,2}) v(\zeta_1, 2) v(\zeta_2, D).\nonumber
\end{gather}
Similarly as in \eqref{22a}, for $(\zeta_1, \zeta_2, \zeta_{12}) \in \mathds{C}^{-}_2$, we have $\Phi_2  (\zeta_1, \zeta_2, \zeta_{12})= Q_2 (z_1, z_2, {\sf J})$, with $\ell_{2,2} (z_1, z_2 ) = (\zeta_1, \zeta_2 , \zeta_{12})$, $(z_1, z_2)\in L^2$, and, this time,
\[
Q_2 (z_1, z_2, 0) = v(z_1^2, 2) v(z_2^2, D). 
\]
Since $ v(\cdot, 2),  v(\cdot, D) \in \mathcal{L}$, by Proposition \ref{2pn}, we then get $\Phi_{2} (\zeta_1, \zeta_2, \zeta_{12})$ for $\zeta_k \in \mathds{C}^{-}$, $k=1,2$. 
For fixed $\zeta_2 \in \mathds{C}^{-}$ and $\zeta_{12}\in \mathds{C}$, the function $\zeta_1 \mapsto \Phi_2  (\zeta_1, \zeta_2, \zeta_{12})$ is in $\mathcal{L}$. By \eqref{L} and \eqref{23}, this yields 
$\Psi_{2,D} (\zeta_1, \zeta_2, \zeta_{12}) \neq 0$  whenever $(\zeta_1, \zeta_2, \zeta_{12})\in \mathds{C}^{-}_2$. This completes the proof of the lemma for $N=2$.

Now we assume that $\Psi_{N-1,D}$ has the stated property. By the second line in \eqref{18}, similarly as in \eqref{23} we get
\begin{eqnarray}
\label{24}
\Psi_{N,D} (\zeta_1, \zeta_2, \zeta_{12}) & = & \left[ (4\pi \partial_1)^m \exp\left( J \Delta_{3,2} \right) v(\zeta_1, 2) \Psi_{N-1,D} (\zeta_2, \zeta_3, \zeta_{23})\right]_{2=3}\\ \nonumber 
& = & \left[ (4\pi \partial_1)^m \widehat{\Phi} (\zeta_1, \zeta_2, \zeta_3,\zeta_{12}, \zeta_{13}, \zeta_{23})\right]_{2=3}\\
& = & (4\pi \partial_1)^m \Phi_3 (\zeta_1, \zeta_2, \zeta_{12}), \nonumber
\end{eqnarray}
where 
\begin{eqnarray}
\label{25}
\Phi_{3} (\zeta_1, \zeta_2, \zeta_{12}) & = & \left[\widehat{\Phi} (\zeta_1, \zeta_2, \zeta_3,\zeta_{12}, \zeta_{13}, \zeta_{23})\right]_{2=3}.
\end{eqnarray}
By \eqref{24}, \eqref{20a}, and \eqref{9}, we obtain that, for $\zeta_{jk} = z_j \cdot z_k$, $1\leq j\leq k \leq 3$, the following holds
\[
\widehat{\Phi} (\zeta_1, \zeta_2, \zeta_3,\zeta_{12}, \zeta_{13}, \zeta_{23}) = Q_{3} (z_1, z_2, z_3, {\sf J})
\]
with ${\sf J} = (J_{jk})$ such that $J_{12} = J$ and $J_{jk} = 0$ otherwise. At the same time, by \eqref{24}, it follows that
\[
Q_{3} (z_1, z_2, z_3, 0) =  v(z_1, 2) \Psi_{N-1,D} (z_2^2, z_3^2, z_{2} \cdot z_{3}).
\]
By the inductive assumption and the fact $v(\cdot, 2)\in \mathcal{L}$, it follows that $Q_{3} (z_1, z_2, z_3, 0)\neq 0$ for  
$(z_1, z_2, z_3)\in L^3$. Then by 
Proposition \ref{2pn}, $\widehat{\Phi}$ does not vanish on $\ell_{3,2}(L^3)\subset M$, see \eqref{16}.   
Since $M \cap \{\{\zeta_{jk}\}: 2=3\}$ is isomorphic to $\mathds{C}^3$, by \eqref{25}, $\Phi_{3}$ does not vanish if $\zeta_k\in \mathds{C}^{-}$, $k=1,2$, which is true for any $\zeta_{12}\in \mathds{C}$. Similarly as above, we fix $\zeta_2 \in \mathds{C}^{-}$ and $\zeta_{12}\in \mathds{C}$, and then conclude that the function $\zeta_1 \mapsto \Phi_2  (\zeta_1, \zeta_2, \zeta_{12})$ is in $\mathcal{L}$. Therefore, 
\begin{equation}
 \label{26}
 \Psi_{N,D}(\zeta_1, \zeta_2, \zeta_{12})\neq 0, \quad
{\rm whenever} \quad  \zeta_1, \zeta_2 \in \mathds{C}^{-}.
\end{equation}
This completes the proof.
\end{proof}
{\it Proof of Theorem \ref{1tm}.} 
Define $\varphi_{N,D} (\zeta) =\Psi_{N,D}(\zeta, \zeta, \zeta)$. By \eqref{26}, it follows that $\varphi_{N,D} \in \mathcal{L}$ for all $N$ and all even $D$. At the same time, by \eqref{11} and \eqref{17}, we have 
\[
Z_N (z) = F_{N} (z,z) = \Psi_{N,D} (z^2, z^2, z^2) = \varphi_{N,D} (z^2),
\]
which by \eqref{4a} yields the proof of the theorem. The case of $m>0$, see \eqref{4a}, is excluded by the fact that $Z_N (0) >0$, while 
$\alpha =0$ follows by growth restrictions. \hfill $\Box$ 

Let us now make some concluding remarks. 
\begin{itemize}
    \item[(a)] In Theorem \ref{1tm}, we assume that the measure $\chi$ is strongly isotropic -- not just isotropic. The reason is to obtain the possibility to deal with 
the family of such measures for all $\mathds{D}$, $D\geq 2$, related to each other by the same $\tau$, cf. \eqref{3b} and \eqref{3c}.
\item[(b)] The choice of $\mathds{Z}$ as the underlying graph is caused by the fact that in the recurrence in \eqref{12} we can deal only with a single ``offspring" of vertex $N$, which is $N-1$. This restriction comes from the condition ${\rm rank}(\Gamma^{m,D}) \leq D$, see \eqref{15z}, which should hold for $D=2$, cf. Proposition \ref{2pn}.

\end{itemize}


\begin{thebibliography}{ll}

\bibitem{Albev} S. Albeverio, Y. Kondratiev, Y. Kozitsky and M.R\"ockner, The Statistical Mechanics of Quantum Lattice Systems: a Path Integral Approach (Vol. 8). European Mathematical Society. 2009.


\bibitem{Bleher} P. M. Bleher, Integration of functions in a space with complex number of dimensions, {\it Theor. Math. Phys.} {\bf 50}, 243--251 (1982). 
 
\bibitem{BM} P. M. Bleher and P. Major, Critical phenomena and universal exponents in statistical physics. On Dyson's hierarchical model, {\it Annals Probab.} {\bf 15},  431--477 (1987).

\bibitem{Borcea} J. Borcea and P. Br\"and\'en, The Lee-Yang and P\'olya-Schur programs. I. Linear operators preserving stability, {\it Inventiones Mathematicae} {\bf 177} 541--569 (2009), II. Theory of stable polynomials and applications, {\it Commun. Pure Appl. Math.} {\bf 62} 1595--1631 (2009).

\bibitem{Borcea1} J. Borcea and P. Br\"and\'en, P\'olya-Schur master theorems for circular domains and their boundaries, {\it Annals of Mathematics} 465--492 (2009).

 \bibitem{Br} P. Br\"and\'en and J Huh, Lorentzian polynomials, {\it Annals of Mathematics} {\bf 192}, 821--891 (2020). 

\bibitem{DN} F. Dunlop and C. M. Newman, Multicomponent field theories and classical
rotators, {\it Commun. Math. Phys.} {\bf 44} 223--235 (1975). 

\bibitem{Frol} J. Fr\"ohlich and P.-M. Rodriguez, Some applications of the Lee-Yang theorem, {\it J. Math. Phys.} {\bf 53}, 095218 (2012). 

\bibitem{Glimm} J. Glimm and A. Jaffe, Correlation inequalities and the Lee-Yang theorem, in: Quantum Physics: A Functional Integral Point of View. New York, NY: Springer New York, 1987. 56-72.


\bibitem{Iliev} L. Iliev,  Laguerre Entire Functions. Publishing House of the Bulgarian Academy of Sciences, Sofia,  1987.


\bibitem{Koz} Y. Kozitsky and N. Melnik, Random vectors with the Lee-Yang property, {\it Theor. Math. Phys.} {\bf 78}, 127--134 (1989). 

\bibitem{Koz1} Y. Kozitsky, Hierarchical model of a vector ferromagnet. Self-similar block-spin distributions and the Lee-Yang theorem, {\it Rep. Math. Phys.} {\bf 26}, 429--225 (1988).

\bibitem{Koz2} Y. Kozitsky, Hierarchical ferromagnetic vector spin model possessing the Lee-Yang property. Thermodynamic limit at the critical point and above, {\it J. Statist. Phys.} {\bf 87} 799--820  (1997).
 \bibitem{Kurtze} D. A. Kurtze, The Yang-Lee edge singularity in one-dimensional 
Ising and N-vector models, {\it J. Statist. Phys.} {\bf 30}, 15--35 (1983). 

\bibitem{LY}  T. D. Lee and C. N. Yang, Statistical theory of equations of state and phase transitions. II. Lattice 
gas and Ising model, {\it Phys. Rev.} {\bf 87} 410--419  (1952). 

\bibitem{LS} E. H. Lieb and A. D. Sokal, A general Lee-Yang theorem for one-component and
multicomponent ferromagnets, {\it Commun. Math. Phys.} {\bf  80}, 153--179 (1981).

\bibitem{Newman} C. M. Newman, Inequalities for Ising models and field theories
which obey the Lee-Yang theorem, {\it Commun. Math. Phys.} {\bf  41}, 1--9 (1975).

\bibitem{Newm} C. M. Newman, Fourier transforms with only real zeros, {\it Proceedings of the American Mathematical Society} {\bf 61(2)}, 245--251 (1976).

 
\bibitem{Simon} B. Simon, Phase Transitions in the Theory of Lattice Gases, Cambridge University Press, 2026. 

\bibitem{Weil} H. Weil, The Classical Groups, their Invariants and Representations, Princeton University Press, 1964. 

\end{thebibliography}
\end{document}